\documentclass{article}

\usepackage{arxiv}

\usepackage[utf8]{inputenc} % allow utf-8 input
\usepackage[T1]{fontenc}    % use 8-bit T1 fonts
\usepackage{hyperref}       % hyperlinks
\usepackage{url}            % simple URL typesetting
\usepackage{booktabs}       % professional-quality tables
\usepackage{amsfonts}       % blackboard math symbols
\usepackage{nicefrac}       % compact symbols for 1/2, etc.
\usepackage{microtype}      % microtypography
\usepackage{lipsum}		% Can be removed after putting your text content
\usepackage{graphicx}
\usepackage{natbib}
\usepackage{doi}

\usepackage{amsmath, amssymb}
\usepackage{enumitem}
\usepackage{caption}
\usepackage{subcaption}
\usepackage[]{graphicx}
\usepackage[ruled,linesnumbered,noend]{algorithm2e} 
\usepackage{mathtools}
\usepackage{bbm}
\usepackage{amsthm}
\usepackage{multirow}
\usepackage{dsfont}
\usepackage{thm-restate}

\newcommand{\prob}{\mathds{P}}

\newcommand{\gaussbm}{\mathsf{GaussBM}}

\newcommand{\punctual}{\mathsf{PunctualSketch}}
\newcommand{\punctualcms}{\mathsf{PunctualCMS}}
\newcommand{\punctualcs}{\mathsf{PunctualCS}}
\newcommand{\lazy}{\mathsf{LazySketch}}
\newcommand{\lazycms}{\mathsf{LazyCMS}}
\newcommand{\lazycs}{\mathsf{LazyCS}}

\newcommand{\cms}{\mathsf{CMS}}
\newcommand{\cs}{\mathsf{CS}}

\newcommand{\mg}{\mathsf{MisraGries}}
\newcommand{\lazyhh}{\mathsf{LazyHH}}

\newtheorem{theorem}{Theorem}
\newtheorem{lemma}{Lemma}
\newtheorem{problem}{Problem}
\newtheorem{fact}{Fact}
\newtheorem{definition}{Definition}

\title{Faster Differentially Private Sketches under Continual Observation}

\author{{Rayne Holland} \\
	Data61\\
	CSIRO\\
	Australia \\
	\texttt{rayne.holland@data61.csiro.au} \\
}

\renewcommand{\shorttitle}{Faster Differentially Private Sketches}

\begin{document}
\maketitle

\begin{abstract}
Linear sketches are  fundamental tools in data stream analytics.
They are notable for supporting both approximate frequency queries and heavy hitter detection with bounded trade-offs for error and memory. 
Importantly, on streams that contain sensitive information, linear sketches can be easily privatized with the injection of a suitable amount of noise.
This process is efficient in the single release model, where the output is released only at the end of the stream.  
In this setting, it suffices to add noise to the sketch once.

In contrast, in the continual observation model, where the output is released at every time-step, fresh noise needs to be added to the sketch before each release.
This creates an additional computational overhead.
To address this, we introduce Lazy Sketch, a novel differentially private  sketching method that employs lazy updates, perturbing and modifying only a small portion of the sketch at each step.
Compared to prior work, we reduce the update complexity by a factor of $O(w)$, where $w$ is the width of the sketch.
Experiments demonstrate that our method increases throughput by up to 250× over prior work, making continual observation differential privacy practical for high-speed streaming applications.

\end{abstract}

% keywords can be removed

\section{Introduction}
\label{sec:intro}

Estimating item frequencies and detecting high-frequency items (called \textit{heavy hitters}) are foundational problems in data stream processing, where algorithms must operate under tight memory and computational constraints.
A prevalent and effective approach to these problems is \emph{linear sketching}, which compresses data by hashing item frequencies into a small array of shared counters.
This technique supports efficient approximate frequency queries while offering strong guarantees on both space usage and estimation accuracy.
Notable examples include Count-Min Sketch ($\cms$)~\cite{cormode2005improved} and Count Sketch ($\cs$)~\cite{charikar2002finding}, which have become standard tools in a range of applications, such as network traffic measurement \cite{bu2010sequential,jang2020sketchflow,liu2016one,li2022stingy,roughan2006secure, yu2013software}.

In privacy-sensitive settings, such as network activity monitoring, it is crucial that streaming data structures avoid leaking sensitive information.
To address this challenge, differential privacy has become the prevailing standard for protecting data streams.
It guarantees that an observer analyzing the outputs of the algorithm is fundamentally limited, in an information-theoretic sense, in their ability to infer the presence or absence of any individual data point in the stream. 

Differential privacy can be provided under two distinct models: \textit{single release} or \textit{continual observation}.
In the single release setting, the output is released only after processing the entire stream.
Linear sketches can be privatized efficiently in this model, by injecting calibrated noise into the sketch once, at the beginning of the stream~\cite{pagh2022improved, zhao2022differentially}. 
This approach achieves strong privacy-utility trade-offs with minimal computational overhead.

The continual observation model demands that an algorithm release outputs at every time step, while safeguarding privacy across the full sequence of releases.
This requirement mirrors practical settings where systems must deliver real-time insights without compromising sensitive data.
Epasto \textit{et al.} recently initiated a study of linear sketches in this setting \cite{epasto2023differentially}. 
Their approach, which we term the Punctual Sketch ($\punctual$), injects fresh noise to every counter at each update. 
This meticulous protection ensures strong privacy guarantees but exacts a heavy computational toll, as the update cost scales with the sketch size.

\subsection{Our Solution}

The $\punctual$ employs a $d \times w$ array of counters, each designed to satisfy differential privacy under continual observation. 
All $dw$ counters are updated at each stream arrival, leading to an update complexity of $\Omega(dw)$.
We adopt the same structural foundation as the $\punctual$, but remove the dependence on $w$ in the update complexity.

Our central innovation lies in the use of lazy updates -- a partial update strategy that modifies and injects noise into only a small, rotating subset of the output sketch at each time step, leaving the remainder untouched. 
This approach amortizes the cost of noise injection over time, resulting in substantial gains in update efficiency.
The outcome is significant as, in practice, applications will increase $w$ to reduce query error.
Therefore, our approach, the $\lazy$, can accommodate increases to $w$ without sacrificing on throughput. 
 
Our contributions are as follows
\begin{itemize}
    \item 
    $\lazy$ (Algorithm~\ref{alg:circular_sketch}), a framework for differentially private linear sketches under continual observation.
    For sketches of width $w$, it observes a $O(w)$ factor reduction in update time against the state-of-the-art, with no loss in utility.
    \item Implementations of continual observation sketches\footnote{\url{https://github.com/rayneholland/CODPSketches}} and, to the best of our knowledge, the first empirical evaluation of this problem.
    \item Experiments show that the $\lazy$ observes up to a factor 250 improvement on throughput against prior work.
    It also systematically observes better utility across a range of parameter settings.
\end{itemize}
\noindent
A high level comparison of our method against prior work is provided in Table~\ref{tab:comparison}.

\begin{table*}[t]
    \centering
    \begin{tabular}{|c||c|c|c|} \hline
        & Accuracy
                & Time
                    & Memory \\ \hline \hline
    \multicolumn{4}{|c|}{\textbf{Frequency Estimation}} \\ \hline
    $\punctual$ \cite{epasto2023differentially}
        & $\gamma_1(T) + \gamma_2(T)$
                & $wd \cdot \mathcal{T}(T)$
                    & $wd \cdot \mathcal{M}(T)$\\ \hline
    $\lazy$ 
        & $\gamma_1(T) + \gamma_2(T/w) + w$
                & $d \cdot \mathcal{T}(T/w)$
                    & $wd \cdot \mathcal{M}(T/w)$\\ \hline
    \end{tabular}
    \caption{
    Big-$\mathcal{O}$ notation is omitted for clarity.
    Accuracy and throughput for continual observation sketches with dimensions $d \times w$.
    $\gamma_1(t)$ refers to the additive error of a \textit{non-private} $d\times w$ linear sketch at time $t$ (that is, the error due to collisions).
    $\gamma_2(t)$ refers to the additive error of $dw$ differentially private continual observation counters at time $t$.
    $\mathcal{T}(t)$ refers to the update time of a continual observation counter on a stream of length $t$.
    $\mathcal{Q}(t)$ refers to the query time of a continual observation counter on a stream of length $t$.
    $\mathcal{M}(t)$ refers to the memory allocation of a continual observation counter on a stream of length $t$.
    Note that, if $\mathcal{M}(t)$ is increasing, for given sketch dimensions $d\times w$, the memory allocation of the $\lazy$ is smaller that that of $\punctual$.
    }
    \label{tab:comparison}
\end{table*}

\subsection{Paper Outline}

Section~\ref{sec:prelims} provides preliminaries that span differential privacy and the fundamentals of linear sketching.
Section~\ref{sec:gauss_bm} introduces a method for counting under continual observation differential privacy.
This method serves as building block in the linear sketches.
Section~\ref{sec:lazy_sketch} presents the $\lazy$, our framework for linear sketching under continual observation differential privacy.
Lastly, Section~\ref{sec:experiments} provides an empirical evaluation of our methods.

\section{Preliminaries}
\label{sec:prelims}
\subsection{Privacy}

Differential privacy ensures that the inclusion or exclusion of any individual in a dataset has a minimal and bounded impact on the output of a mechanism.
Two streams $\mathcal{X} = \{x_1, \ldots, x_T\}$ and $\mathcal{X}^{\prime}= \{x_1^{\prime}, \ldots, x_T^{\prime}\}$ are \textit{neighboring}, 
denoted $\mathcal{X} \sim \mathcal{X}^{\prime}$,
if they differ in one element.
Formally, $\mathcal{X} \sim \mathcal{X}^{\prime}$ if there exists a unique $i$ such that
$x_i \neq x_i^{\prime}$.
The following definition of differential privacy is adapted from Dwork and Roth~\cite{dwork2014algorithmic}.
\begin{definition}[\((\varepsilon, \delta)\)-Differential Privacy]
Let $\mathcal{X}, \mathcal{X}' \in \mathcal{U}^*$ be neighboring input streams (i.e., differing in at most one element). A randomized mechanism $\mathcal{M}$ satisfies \((\varepsilon, \delta)\)-differential privacy under a given observation model if for all measurable sets of outputs \( Z \subseteq \mathtt{support}(\mathcal{M}) \), it holds that:
\[
\prob[\mathcal{M}(\mathcal{X}) \in Z] \leq e^{\varepsilon} \cdot \prob[\mathcal{M}(\mathcal{X}') \in Z] + \delta.
\]
\label{def:diff_privacy}
\end{definition}
\noindent
In streaming there are two types of observation model: \textit{single release} and \textit{continual observation}.
In the single release model, the mechanism processes the stream once and releases a single differentially private output after the final stream update has been observed.
Under continual observation, the mechanism releases an output at each time step, and differential privacy is required to hold for the sequence of outputs up to every time step.

The Gaussian mechanism is a fundamental technique for ensuring $(\varepsilon, \delta)$-differential privacy by adding noise calibrated to a function's $\ell_2$-sensitivity.  
Let $\triangle_2(f) = \max_{\mathcal{X} \sim \mathcal{X}^{\prime}} \lVert f(\mathcal{X}) - f(\mathcal{X}^{\prime}) \rVert_2$ denote the $\ell_2$-sensitivity of the function $f$.

\begin{lemma}[Gaussian Mechanism {\cite{dwork2014algorithmic}}]
Let \( f \) be a function with $\ell_2$-sensitivity $\triangle_2(f)$. For any $\varepsilon \in (0,1)$ and $\delta \in (0,1)$, the mechanism
\[
\mathcal{M}(\mathcal{X}) = f(\mathcal{X}) + \mathcal{N}\left(0, \sigma^2 I\right),
\]
where $\sigma \geq \frac{\triangle_2(f) \cdot \sqrt{2 \ln (1.25/\delta)}}{\varepsilon}$, satisfies $(\varepsilon, \delta)$-differential privacy.
\label{lem:gaussian_mechanism}
\end{lemma}

\begin{fact}[Tail bound on the sum of Gaussian RVs]
Let $Y = \sum_{i=1}^h Y_i$, where $Y_i \sim \mathcal{N}(0, \sigma^2)$ are i.i.d. Gaussian random variables. Then, for any $\beta \in (0,1)$,
\[
\Pr\left[|Y| \geq \sigma \sqrt{2 h \ln(2/\beta)}\right] \leq \beta.
\]
\label{fact:gaussian_sum}
\end{fact}

\subsection{Linear Sketching }

Linear sketches enable approximate frequency queries, in small memory, by hashing items to a compact collection of shared counters. 
We focus on the well known Count-min Sketch ($\cms$) \cite{cormode2005improved} and Count Sketch ($\cs$) \cite{charikar2002finding}. 
While they share a similar structure, they differ in their update and query procedures and offer different error trade-offs.

\subsubsection*{Count-Min Sketch}
The $\cms$ uses $d$ hash functions $h_1, \ldots, h_d : \mathcal{U} \rightarrow [w]$ to maintain a $d \times w$ matrix of counters $C \in (\mathrm{R}^+)^{d \times w}$. 
For each incoming item $x \in \mathcal{U}$ with frequency increment $c_x$, the update rule is:
\[
    C[i, h_i(x)] \leftarrow C[i, h_i(x)] + c_x, \quad \text{for all } i \in [d].
\]
This update procedure is illustrated in Figure~\ref{fig:cms}.
To estimate the frequency of $x$, $\cms$ returns $\hat{f}_x = \min_{i \in [d]} C[i, h_i(x)]$.
This estimator is biased upwards due to hash collisions with other items but offers strong probabilistic error guarantees.

\begin{lemma}[Approximation error for $\cms$ \cite{cormode2005improved}]
    For $x \in \mathcal{U}$ and any $\beta, \eta \in (0,1)$, the estimation error of Count Min Sketch with  \( d = \lceil \log(1/\beta) \rceil \)  repetitions and table size $w = \lceil 2/\eta \rceil $ satisfies,
    \begin{align*}
        \prob[\hat{f}_x - f_x > \eta F_1] \leq \beta, 
    \end{align*}
    where $F_1 = \sum_{y\in \mathcal{U}} f_y$.
    \label{lem:cms_approx}
\end{lemma}

\subsubsection*{Count Sketch}
The $\cs$  modifies the $\cms$ by incorporating randomized signs to reduce bias. 
In addition to the hash functions $h_1, \ldots, h_d : \mathcal{U} \rightarrow [w]$, it uses a second family of hash functions $g_1, \ldots, g_d : \mathcal{U} \rightarrow \{-1, +1\}$ that assign random signs.
For each $x \in \mathcal{U}$, with increment $c_x$, the update rule is:
\[
    C[i, h_i(x)] \leftarrow C[i, h_i(x)] + g_i(x) \cdot c_x, \quad \text{for all } i \in [d].
\]
To estimate the frequency of $x$, $\cs$ returns:
\[
    \hat{f}_x = \mathrm{median}_{i \in [d]} \left( g_i(x) \cdot C[i, h_i(x)] \right).
\]
This estimator is unbiased and has error bounded in terms of the  second moment of the frequency vector.
\begin{lemma}[Approximation error for $\cs$ \cite{charikar2002finding}]
    For \( x \in [n] \) and any \( \beta, \eta \in (0,1) \), the estimation error of Count Sketch with \( d = \lceil \ln(1/\beta) \rceil \) and table width \( \lceil w = 3/\eta^2 \rceil \) satisfies:
    \[
        \Pr\left[ |\hat{f}_x - f_x| > \eta\sqrt{F_2/} \right] \leq \beta,
    \]
    where \( F_2 = \sum_{y \in [n]} f_y^2 \) is the second moment of the frequency vector.
    \label{lem:cs_approx}
\end{lemma}

\begin{figure}
    \centering
    \includegraphics[width=0.6\linewidth]{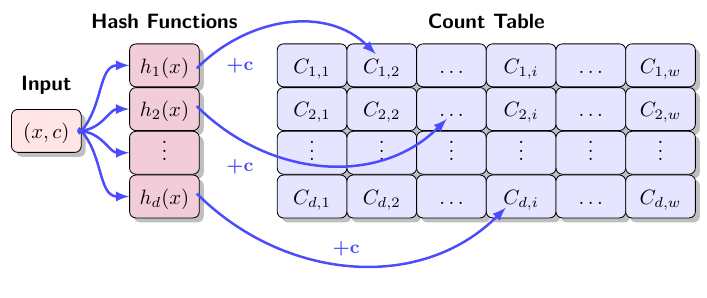}
    \caption{
    The \textsf{update} procedure for a Count-Min Sketch. 
    It follows that~$h_j(x)=i$.
    }
    \label{fig:cms}
\end{figure}

\section{The Gaussian Binary Mechanism}
\label{sec:gauss_bm}

We introduce the Gaussian Binary Mechanism ($\gaussbm)$, a variation of the standard binary mechanism that replaces Laplace noise with Gaussian noise. 
This mechanism serves as the core tool for solving the $n$-counters problem -- a generalization of the standard summing problem \cite{andersson2023smooth,fichtenberger2023constant,henzinger2024unifying}.
We will use the $n$-counters problem to formalize sketching in the continual observation model as a structured counting task. 
Lastly, we derive concrete utility bounds for the collection of counters, which will support our approach to identifying heavy hitters under continual observation.

Continual observation requires injecting noise at every update, rather than just once. 
This raises the challenge of minimizing cumulative error over time.
A standard approach is the \textit{binary mechanism} \cite{chan2011private,dwork2010differential}, which organizes the stream of updates into a binary tree structure.
Each update is represented by a leaf node and internal nodes correspond to the intervals comprised of all descendant leaf nodes. 
An illustration of this binary tree structure is given in Figure~\ref{fig:gauss_bm}.
A query for the value in the counter at time $t$ can be constructed from the tree by aggregating the noisy counts of intervals that partition $[1,t]$.
This reduces noise relative to a naive approach that is equivalent to summing all the noisy counts in the leaf nodes. 

The binary mechanism with Gaussian noise, detailed in Algorithm~\ref{alg:binary_mechanism}, operates by storing the increment $c$ at each time step $t$ as a leaf node with fresh Gaussian noise. 
It then merges tree nodes whose time intervals form complete binary intervals ending at $t$, summing their counts and adding new Gaussian noise at each merge.

\begin{algorithm}[t]
  \SetAlgoLined
  \DontPrintSemicolon
  \SetKwFunction{init}{initializeTree}
  \SetKwFunction{update}{updateTree}
  \SetKwFunction{query}{getNoisySum}
  \SetKwProg{myproc}{define}{}{}
  \SetKwInOut{require}{Require}

  \myproc{$\gaussbm(T, \varepsilon, \delta, m)$}{
        $\mathtt{tree} \gets [\,]$ \;
        $h \gets \lceil \log(T + 1) \rceil$ \;
        $\sigma \gets \varepsilon^{-1} \sqrt{2hm\ln(1.25/\delta)}$ \;
        \KwRet
    }
  
  \myproc{$\mathsf{update}(c, t)$}{
        Append $([t,t], c, \mathcal{N}(0, \sigma))$ to $\mathtt{tree}$ \;
        $l \gets 0$ \;
        \While{$|\mathtt{tree}| > 1$ \textbf{and} $t \,\&\, (1 \ll l)$}{
            $(I_r, c_r, \_) \gets \mathtt{tree}.\mathsf{pop}()$ \;
            $(I_l, c_l, \_) \gets \mathtt{tree}.\mathsf{pop}()$ \;
            $I \gets I_l \cup I_r$; $c \gets c_l + c_r$ \;
            Append $(I, c, \mathcal{N}(0, \sigma))$ to $\mathtt{tree}$ \;
            $l \gets l + 1$ \;
        }
    }

  \myproc{$\mathsf{query}()$}{
        $\mathtt{total} \gets 0$ \;
        \ForEach{$(\_, c, \gamma) \in \mathtt{tree}$}{
            $\mathtt{total} \gets \mathtt{total} + c + \gamma$ \;
        }
        \KwRet $\mathtt{total}$ \;
    }

  \caption{Gaussian Binary Mechanism}
  \label{alg:binary_mechanism}
\end{algorithm}

\begin{figure}
    \centering
    \includegraphics[width=0.6\linewidth]{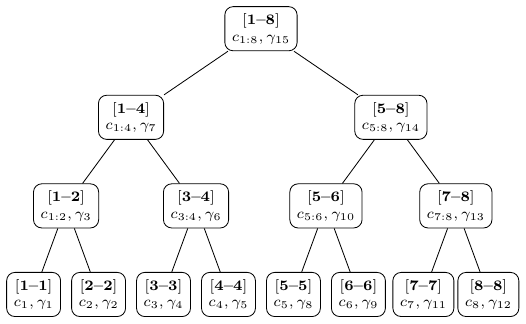}
    \caption{
    The Binary Mechanism. 
    $c_i$ denotes the counter increment at time $i$.
    The sum released at time $t=7$ is $c_7+c_{5:6}+c_{1:4}+ \gamma_7 + \gamma_{10} + \gamma_{11}$.
    }
    \label{fig:gauss_bm}
\end{figure}

\subsection{The $n$-Counters Problem}

Multiple copies of the binary mechanism can be used to provide approximate differential privacy for a collection of counters.
The following formalizes the problem of releasing a collection of counters under a stream of updates.
\begin{problem}
[$n$-Counters]
Let $T \in \mathbb{N}$ be the number of time steps, and let $n \in \mathbb{N}$ be the number of counters. 
The input is a sequence of update vectors $V = (\mathbf{v}_1, \ldots, \mathbf{v}_T)$, where each $\mathbf{v}_t \in \mathbb{Z}^n$ specifies the increment to each counter at time $t$. 
At each time step $t \in [T]$, the goal is to release the prefix sums for each counter:
\[
x_t^{(i)} = \sum_{j=1}^t v_j^{(i)} \quad \text{for all } i \in [n].
\]
\label{prob:m_counters}
\end{problem}

In order to formalize differential privacy for this problem, we introduce a notion of neighboring inputs.
We include some flexibility in the definition and permit a parameterized number of counters to differ.

\begin{definition}
[$m$-Neighboring Inputs]
Two input sequences $V = (\mathbf{v}_1, \ldots, \mathbf{v}_T)$ and $V' = (\mathbf{v}'_1, \ldots, \mathbf{v}'_T)$ in $\mathbb{Z}^{n \times T}$ are said to be \emph{$m$-neighboring}, denoted $ V \overset{m}{\sim} V'$, if there exists a subset $I \subseteq [n]$, with $|I| \leq m$, such that:
\[
\sum_{t=1}^T |v_t^{(i)} - v_t^{\prime(i)}| \leq 1 \quad \text{for all } i \in I,
\]
and
\[
\sum_{t=1}^T |v_t^{(j)} - v_t^{\prime(j)}| = 0 \quad \text{for all } j \in [n]\setminus I,
\]
\label{def:m_neighbors}
\end{definition}
\noindent 
This definition says that at most $m$ counters can differ at a single increment, and that this difference is bounded by one.

The sensitivity of the input streams will be used to scale the variance of the noise injected into the $\gaussbm$ counters.
Following Definition~\ref{def:m_neighbors}, the $\ell_2$ sensitivity of $m$-neighboring inputs $V\overset{m}{\sim} V^{\prime}$ is:
\begin{align*}
    \Delta_2 
    &= \max_{V \overset{m}{\sim} V' } \left( \sqrt{ \sum_{i=1}^m \sum_{t=1}^T \left( v_t^{(i)} - v_t^{(i)\prime} \right)^2 } \right)
    \leq \sqrt{m}.
\end{align*}
The solution to Problem~\ref{prob:m_counters}, using $\gaussbm$ counters to store noisy counts, is given in Algorithm~\ref{alg:m_count_gbm}.
This algorithm has the following privacy gaurantee.

\begin{algorithm}[t]
  \SetAlgoLined
  \DontPrintSemicolon
  \SetKwFunction{algo}{initialise}\SetKwFunction{proc}{update}\SetKwFunction{out}{query}
  \SetKwProg{myproc}{define}{}{}
  \SetKwInOut{require}{Require}
  \myproc{$\mathsf{PrivateCounters}(V,n, m, \varepsilon, \delta)$}{
        \For{$i \in [n]$}
        {
            $O_i \gets \gaussbm(|V|, \varepsilon, \delta, m)$\; 
        }
        \For{$\mathbf{v}_t \in V$}
        {
            \For{$i \in [n]$}
            {
                $O_i.\mathsf{update}(v_t^{(i)})$\;
            }
            \KwRet $\{ O_i \mid i \in[n]\}$\;
        }
        
        }{}
  \caption{$n$-Counters with $\gaussbm$}
  \label{alg:m_count_gbm}
\end{algorithm}

\begin{theorem}[Gaussian Binary Mechanism]
    On input $V$, and for any $\varepsilon > 0$ and $ \delta \in (0,1)$, 
    Algorithm~\ref{alg:m_count_gbm} is $(\varepsilon, \delta)$-differentially private, for $m$-neighboring inputs, in the continual release model. 
    \label{thm:gaussian_counters}
\end{theorem}
\begin{proof}
    Each counter is implemented with the Gaussian Binary Mechanism.
    Over the course of the stream, each mechanism produces a binary tree with a noisy count of every interval of the stream.
    These binary trees can be processed to produce the exact output sequence of Algorithm~\ref{alg:m_count_gbm}.
    
    Now, suppose the function $\mathcal{F}$ takes $V$ as input and outputs the complete binary tree of each counter stream, with \textit{exact} counts at each tree node.
    On $m$-neighboring inputs $V \overset{m}{\sim} V'$, there are at most $m$ trees that differ in the outputs of $\mathcal{F}(V)$ and $\mathcal{F}(V^{\prime})$.
    We refer to these trees as \textit{neighboring}.
    Each neighboring tree differs in one leaf node and at all internal nodes from this leaf to the root.
    Recall that an internal node aggregates the counts of all descendant leaf nodes.
    By Definition~\ref{def:diff_privacy}, the exact counts of the nodes of neighboring trees differ by at most 1.
    Therefore, as the height of the binary tree is $h= \lceil \log (|V|+1) \rceil$, 
    the $\ell_2$ sensitivity of $\mathcal{F}$, on $m$-neighboring inputs, is $\Delta_2^{\mathcal{F}} =\sqrt{m h}$.

    Let $\mathcal{M}(V, \varepsilon, \delta)$ denote a mechanism that computes $\mathcal{F}(V)$ and adds $ \mathcal{N}(0, \sigma)$, with $\sigma =\varepsilon^{-1}\Delta_2^{\mathcal{F}} \sqrt{\ln(1.25/\delta)}$, to each node count.
    By Lemma~\ref{lem:gaussian_mechanism}, $\mathcal{M}$ is $(\varepsilon, \delta)$-differentially private for $m$-neighboring inputs.
    As the outputs of Algorithm~\ref{alg:m_count_gbm} can be generated from $\mathcal{M}$ as a post-processing step, it  is also $(\varepsilon, \delta)$-differentially private.
\end{proof}

We now consider the maximum additive error out of $n$ $\gaussbm$ counters.
\begin{lemma}
    Let $V$ be a stream of updates to $n$ counters, where $|V| = T$ and $m \in \mathbb{Z}^+$ denotes the sensitivity of the inputs. 
    For any privacy parameters $\varepsilon > 0$ and $\delta \in (0,1)$, let the counters be initialized using $\gaussbm(|V|, \varepsilon, \delta, m)$. 
    Then, for any $\beta \in (0, 0.5)$ and $I \subseteq [T]$, with probability at least $1 - \beta$, the maximum additive error of any counter at any timestep $t \in I$ is bounded by
    \[
        \frac{1}{\varepsilon} (\log T + 1) \sqrt{2m \ln\left(\frac{1.25}{\delta}\right) \ln\left(\frac{2|I|n}{\beta}\right)}.
    \]
    \label{lem:gbm_utility}
\end{lemma}

\begin{proof}
    Each count released by a Gaussian Binary Mechanism is the sum of at most $h = \lceil \log (T+1) \rceil$ internal node counts with noise $\mathcal{N}(0, \sigma)$.
    By Fact~\ref{fact:gaussian_sum}, a single counter at time $t$, with probability $1-\beta/(|I|n)$, has additive error at most
    \[
        \sigma \sqrt{2h \ln(2 |I| n/ \beta)} = \frac{h}{\varepsilon}\sqrt{2m\ln (1.25/\delta)\ln(2|I|n/\beta)}.
    \]
    Taking a union bound over $n$ counters and $|I|$ time steps, with probability $1-\beta$, this additive error bound holds for all counters at all time steps.
\end{proof}

\section{Lazy Frequency Estimation}
\label{sec:lazy_sketch}

We will first introduce  the Punctual Sketch ($\punctual$) \cite{epasto2023differentially}.
This approach has high computational overhead and acts a straw man.
Subsequently, we will introduce our new framework, $\lazy$, that reduces the computational costs of the $\punctual$.

\begin{algorithm}[t]
  \SetAlgoLined
  \DontPrintSemicolon
  \SetKwFunction{algo}{initialise}\SetKwFunction{proc}{update}\SetKwFunction{out}{query}
  \SetKwProg{myproc}{define}{}{}
  \SetKwInOut{require}{Require}
  \myproc{$\mathsf{PunctualSketch}(X, w,d, \varepsilon, \delta)$}{
        \For{$i \in [d]$}
        {
            Randomly select hash functions $h_i:\Omega\rightarrow [w]$ \;
            For $\cs$ randomly select hash functions $g_i:\Omega \rightarrow \{-1,1\}$\; 
        }
        \For{$(i,j) \in [d] \times [w]$}
        {
            $O[i][j] \gets \mathsf{GaussBM}(|X|, \varepsilon, \delta, 2d)$\;
        }
        \For{$x_t \in X $}
        {
            \For{~$i \in [d]$}
            {
                $h\gets h_i(x)$ \; 
                \For{$j \in [w]$}
                {
                    \If{$j = h$}
                    {
                        $c \gets \begin{cases}
                        g_i(x)  &  \text{for } \cs \\
                         1 & \text{for } \cms
                        \end{cases}$ \;
                    }
                    \Else 
                    {
                        $c \gets 0$
                    }
                    $O[i][j].\mathsf{update}(c, t)$\;
                }
            }
        }
        \KwRet $O$\;
        }{}
  \caption{The $\punctual$}\label{alg:sketch_punctual}
\end{algorithm}

\subsection{The Punctual Sketch}
In contrast to the single release model, the continual observation setting demands a more delicate touch: fresh noise must be added to the sketch after every update, without exception. 
The $\punctual$ meets this stringent requirement using a $d \times w$ array of continual observation counters, denoted $O$.
Upon the arrival of an update $x$, the sketch proceeds as usual, incrementing the counters $O[i][h_i(x)]$ for all $i \in [d]$, in accordance with the underlying sketching technique. 
However, continual observation privacy requires more than selective updates. 
Even those counters not directly touched by $x$ -- that is, all $O[i][j]$ where $j \neq h_i(x)$ -- must also be refreshed, each receiving an update with value zero. 
This ensures that every cell in the sketch is revisited and re-noised at each time step, in a precisely timed, punctual manner.
This update strategy is formalized in Algorithm~\ref{alg:sketch_punctual}.

The $\punctual$ serves as a concrete instance of the $(dw)$-Counters problem (Problem~\ref{prob:m_counters}).
To understand its privacy guarantees, consider two input streams, $X$ and $X'$, that differ only at a single point: the $l^{\text{th}}$ update, where $x_l \neq x_l'$. 
For every time step $t \neq l$, the input to each counter in the array $O$ is identical under both $X$ and $X'$. 
The distinction arises solely at time $l$.

Now, fix a row $i \in [d]$. 
If the hash functions differ on the $l^{\text{th}}$ update --that is, if $h_i(x_l) \neq h_i(x_l')$ -- then the counters $O[i][h_i(x_l)]$ and $O[i][h_i(x_l')]$ receive different updates. 
Specifically, one receives update $\pm 1$, while the other receives a zero. 
This results in a unit difference in the input stream to each affected counter.
Since this discrepancy can occur in each of the $d$ rows, the total input difference across the array amounts to $2d$ neighboring streams. 
To account for this, each counter is instantiated using  $\gaussbm(|X|, \varepsilon, \delta, 2d)$.
By Theorem~\ref{thm:gaussian_counters}, this ensures that the $\punctual$ satisfies $(\varepsilon, \delta)$-differential privacy.

This mechanism of punctual updates to the output counters stands in contrast to our approach, which lazily delays updates to the output to increase throughput.
Although the $\punctual$ is sufficiently private, it is not practical as the update time is at least linear in the width of the sketch.
A goal of this work is to remove this dependence.

The query procedure is presented in Algorithm~\ref{alg:sketch_query}.
It applies the standard query procedures for non-private sketches to the array of continual observation counters.

\begin{algorithm}[t]
  \SetAlgoLined
  \DontPrintSemicolon
  \SetKwFunction{algo}{initialise}\SetKwFunction{proc}{update}\SetKwFunction{out}{query}
  \SetKwProg{myproc}{define}{}{}
  \SetKwInOut{require}{Require}
    \myproc{$\mathsf{query}(x)$}{
        \For{$i \in [d]$}
        {
            $j \gets h_i(x)$\;
            $v_i \gets  \begin{cases}
                    g_i(x) \cdot O[i][j].\mathsf{query}() &  \text{for } \cs \\
                     O[i][j].\mathsf{query}() & \text{for } \cms
                    \end{cases}$
        }  
        \KwRet $\begin{cases}
                   \text{median } \{ v_i \mid i \in [d] \}  &  \text{for } \cs \\
                    \min \{ v_i \mid i \in [d] \} & \text{for } \cms
                \end{cases}$ \;
        }{}
  \caption{CO Sketch Query}\label{alg:sketch_query}
\end{algorithm}

\subsection{The Lazy Sketch}
In the $\punctual$, \textit{every} cell is touched at \textit{every} stream update, recieving a value from $\{-1,0,1\}$.
Therefore, the time complexity for each update is $\Omega(dw)$.
To overcome this dependence on $w$ we utilize \textit{lazy} updates to the output counters.
The key idea is to relax the update schedule by partitioning the stream of inputs to each continual observation counter into intervals. 
During an interval, we accumulate the count for the counter without publishing it to the output.
When an interval ends, we push the count from the current interval to the corresponding output counter.
We refer to this process as the $\lazy$. 

Formally, the $\lazy$ contains a $d\times w$ array of \textit{exact} counters $P$ and a $d\times w$ array of \textit{continual observation} counters $O$.
The array $P$ maintains the exact count for each cell in its current interval and $O$ contains a noisy count for each cell aggregated from all prior intervals.
The array $O$ constitutes the output of the mechanism and $P$ is never revealed to the public.

The mechanism updates $P$ as it would a non-private sketch.
For example, on update $x_t$, using the $\cms$ technique, $\forall i \in [d]$, $P[i][h_i(x_t)] = P[i][h_i(x_t)] + 1$.
These counts are then lazily pushed to the output.
Per update, we push one column from $P$ to the continual observation counters $O$, selected according to a \textit{round robin} scheme.
Formally, for $\mathtt{pos} = t \mod w$, $\forall i \in [d]$, each output counter $O[i][\mathtt{pos}]$ gets updated with the value $P[i][\mathtt{pos}]$.
This constitutes the conclusion of an interval for cells in column $\mathtt{pos}$ and $P[i][\mathtt{pos}]$ gets reset to $0$, $\forall i \in [d]$.
The update procedure is illustrated in Figure~\ref{fig:lazy_sketch}.

\begin{algorithm}[t]
  \SetAlgoLined
  \DontPrintSemicolon
  \SetKwFunction{algo}{initialise}\SetKwFunction{proc}{update}\SetKwFunction{out}{query}
  \SetKwProg{myproc}{define}{}{}
  \SetKwInOut{require}{Require}
  \myproc{$\lazy(X,d,w, \varepsilon, \delta)$}{
        \For{$i \in [d]$}
        {
            Randomly select hash functions $h_i:\Omega\rightarrow [w]$ \;
            For $\cs$ randomly select hash functions $g_i:\Omega \rightarrow \{-1,1\}$\; 
        }
        $\mathtt{pos}\gets 0$\;
        \For{$(i,j) \in [d] \times [w]$}
        {
            $P[i][j] \gets 0$\;
            $O[i][j] \gets \mathsf{GaussBM}(|X|/w,\varepsilon, \delta, 2d)$\;
        }
        \For{$x_t \in X$}
        {
            \For{~$i \in [d]$}
            {
                $c \gets 
                    \begin{cases}
                        g_i(x_t)  &  \text{for } \cs \\
                        1 & \text{for } \cms
                    \end{cases}$ \;
                $P[i][h_i(x_t)] \gets P[i][h_i(x_t)] + c$\; 
                \tcp{update $O$ at $\mathtt{pos}$}
                $O[i][\mathtt{pos}].\mathsf{update}(P[i][\mathtt{pos}], \lfloor t/w\rfloor +1)$\;
                $P[i][\mathtt{pos}] \gets 0$\;
            }
            $\mathtt{pos} \gets \mathtt{pos} + 1 \mod w$\;
            \KwRet $O$\;
        }
        }{}
  \caption{The $\lazy$}\label{alg:circular_sketch}
\end{algorithm}

An update involves $2d$ changes to $P$ and $d$ changes to $O$.
Thus, we drop dependence on $w$ in the update complexity.
The width of an interval for each cell has length $w$.
Updates to the output are lazily delayed by the intervals.
Therefore, there is an additive error of at most $w$ due to lazy updates.
However, as only one cell in each row of $P$ gets incremented at an update, the expected value of the count of a cell in each interval is $1$.
Formally, for $\mathtt{pos} = t \mod w$ and row $i$,
\begin{align*}
    \mathbbm{E}_{h_i}\left[P[i][\mathtt{pos}] \right] 
    &=
    \sum_{q= t-\mathtt{pos}}^{t} \mathbbm{1}[h_i(x_q) = \mathtt{pos}]
    \leq \sum_{q= t-\mathtt{pos}}^{t} \frac{1}{w} = 1.
\end{align*}
%where the randomness is over the choice of $h_i$ from a universal family of hash functions.
Therefore, the additive error observed from the delay is small on average.

In addition, the amount of noise added to the output counters in the $\lazy$ is smaller than the amount of noise added to the output counters in the $\punctual$.
This is because each output counter now receives a stream of updates of length $T/w$.
Thus, noise is only added every $w$ updates.
This is in contrast to the $\punctual$, where noise is added to each cell at every update.

The $\lazy$ utilizes the same query procedure (Algorithm~\ref{alg:sketch_query}) as the $\punctual$.

\begin{figure}
    \centering
    \includegraphics[width=0.6\linewidth]{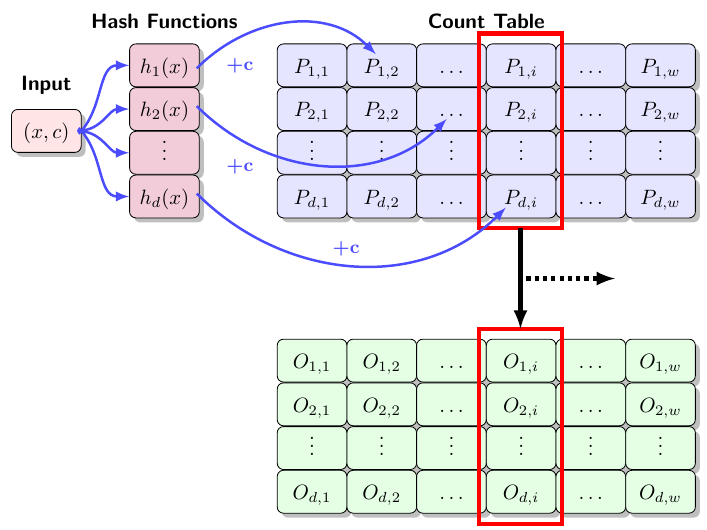}
    \caption{
    Update procedure for the $\lazy$.
    At time $t$, column $i = t\mod w$ from $P$ gets pushed to the output sketch.
    }
    \label{fig:lazy_sketch}
\end{figure}

\subsection{Properties}

We begin with the privacy result.

\begin{theorem}[$\lazy$ Privacy]
    Algorithm~\ref{alg:circular_sketch} is $(\varepsilon, \delta)$-differentially private.
    \label{thm:lazysk_privacy}
\end{theorem}

\begin{proof}
    The goal is to show that Algorithm~\ref{alg:circular_sketch} is an instance of the $(dw)$-counters problem and that input streams to the output counters are $2d$-neighboring.
    The privacy result then follows from Theorem~\ref{thm:gaussian_counters}.
    
    The output counter $O[i][j]$ receives a stream of periodic updates from $P[i][j]$ that occur at times $t$, where $t \mod w = j$.
    Let $S^{(i,j)}= c_1^{(i,j)}, c_2^{(i,j)}, \ldots$ denote the stream of counts that update $O[i][j]$ on input $X$.
    The $k^{\text{th}}$ update to output column $j$ is a function of the stream arrivals on the interval $I_k^{(j)} = (k\lfloor t/w\rfloor +j-w, \ldots, k\lfloor t/w\rfloor +j]$.
    Specifically,
    \[
        c_k^{(i,j)} = \sum_{t \in I_k^{(j)}} \mathbf{1}[h_i(x_t) = j].
    \]
    The set of streams $S=\{S^{(i,j)} \mid (i,j) \in [d]\times [w]\}$ denotes all the updates to the output counters.
    Although the intervals $I_k^{(j_1)}$ and $I_k^{(j_2)}$ are disjoint, for $j_1 \neq j_2$, we can treat $c_k^{(i,j_1)}$ and $c_k^{(i,j_2)}$ as occurring at the same time.
    This is because each output counter receives one update (its $k^{\text{th}}$ update) in the time window $[kw,(k+1)w-1]$.
    Thus, $S$ is an instance of the $(dw)$-counters problem.
    
    Consider two streams $X$ and $X^{\prime}$ that only differ at the $l^{\text{th}}$ update, where $x_l \neq x_l^{\prime}$.
    The cells
    \[
        A = \{ (i,j) \mid (h_i(x_l) = j \lor h_i(x_l^{\prime}) = j) \land h_i(x_l) \neq h_i(x_l^{\prime})\}
    \]
    will experience different input streams from $X$ and $X^{\prime}$.
    Specifically, for $(i,j) \in A $ and $\hat{k}$ such that $l \in I_{\hat{k}}^{(j)}$, $|c_{\hat{k}}^{(i,j)}-c_{\hat{k}}^{\prime (i,j)} | =1$.
    As $|A| \leq 2d$, the sets of streams $S$ and $S^{\prime}$, generated by inputs $X$ and $X^{\prime}$, are $2d$-neighboring.

    By Theorem~\ref{thm:gaussian_counters}, an instance of the $(dw)$-counters problem, on a stream of length $\lceil T/w \rceil $ is $(\varepsilon, \delta)$-differentially private on $2d$-neighboring inputs if each counter is initialized as $\mathsf{GaussBM}(\lceil T/w \rceil,\varepsilon, \delta, 2d)$
    Thus, Algorithm~\ref{alg:circular_sketch} is $(\varepsilon, \delta)$-differentially private.
\end{proof}

With privacy established, we now present the utility and performance of the $\lazy$ under both the $\cms$ and $\cs$ techniques.

\begin{lemma}[$\lazy$ Utility]
    Let $\beta \in (0, 0.5)$ and $\eta \in (0,1)$. 
    On input $X \in \mathcal{U}^T$,
    at all times $t \in [T]$ and for any $a \in \mathcal{U}$, the following guarantees hold with probability at least $1 - \beta$:
    \begin{itemize}
        \item \textbf{$\lazycms$}, with dimensions $d= \lceil \log \tfrac{2T}{\beta} \rceil $ and $w = \lceil 2/\eta \rceil$,  returns an estimate $\hat{f}_{a,t}$ such that:
        \[
         -\gamma - w   \leq \hat{f}_{a,t} - f_{a,t} \leq \eta t + \gamma.
        \]

        \item \textbf{$\lazycs$}, with dimensions $d= \lceil \log \tfrac{2T}{\beta} \rceil $ and $w = \lceil 3/\eta^2 \rceil$,  returns an estimate $\hat{f}_{a,t}$ such that:
        \[
            |\hat{f}_{a,t} - f_{a,t}| \leq \eta\sqrt{F_{2,t}} + \gamma + w.
        \]
    \end{itemize}
    Here, the noise term $\gamma$ is defined as:
    \[
        \gamma = \frac{3\log (T/w)}{\varepsilon} \sqrt{ d \ln\left( \frac{2Td }{\beta} \right) \ln\left( \frac{1.25}{\delta} \right) }.
    \]
    Both algorithms support update and query operations in time $\mathcal{O}\left(\log\tfrac{T}{\beta } \log \tfrac{T}{w}\right)$ and use $\mathcal{O}\left(w \log \tfrac{T}{\beta } \log \tfrac{T}{w}\right)$ words of memory.
    \label{thm:lazy_utility}
\end{lemma}
\begin{proof}
    We start with the error bounds.
    There are three sources of error in the sketch: (1) collisions; (2) lazy updates; and (3) privacy noise.
    We look at each source of error separately and begin with the error due to collisions.
    Algorithm~\ref{alg:circular_sketch} incorporates lazy updates to the output.
    Therefore, the input stream to the output sketch $O$ is a subset of the input stream to a non-private sketch.
    As a result,
    the error due to collisions is bounded by the error in an equivalent non-private sketch.
    
    For an item $a \in \mathcal{U}$,
    let $\hat{f}_{a,t}^{\mathsf{alg}}$ denote the estimate from algorithm $\mathsf{alg}$ at time $t$.
    With sketch depth $d = \ln \tfrac{2T}{\beta}$ and width $w= \lceil e/\eta \rceil$,  by Lemma~\ref{lem:cms_approx}, with probability $1-2\beta/T$, $\lazycms$ has collision error 
    \begin{align*}
        0\leq\hat{f}_{a,t}^{\mathsf{CMS}} - f_a \leq \eta t.
    \end{align*}
    By a union bound, this additive error holds at all times $t\in[T]$ with probability $1-\beta/2$.
    Similarly, by Lemma~\ref{lem:cs_approx}, with probability $1-\beta/2$, $\lazycs$ has collision error
    \begin{align*}
        |\hat{f}^{\mathsf{CS}} - f_a| \leq [-\eta\sqrt{F_2}, \eta\sqrt{F_2}],
    \end{align*}
    at all times $t\in[T]$.

    An output counter is updated every $w$ updates.
    The maximum count of an item $a$ between updates to an output counter is $w$. 
    Therefore, the error due to lazy updates is at least $-w$ and at most $0$ for $\lazycms$.
    As the $\cs$ can have negative increments, the $\lazycs$ observes an additive error of $[-w,w]$ incurred by lazy updates.

    Let $\mathsf{err}_{i,j}$ denote the size of the noise  perturbation in counter $O_{i,j}$.
    By Lemma~\ref{lem:gbm_utility}, with probability $1-\beta/2$, the maximum additive error from $dw$ $\gaussbm$ counters, on $2d$-neighboring streams of length $\lceil T/w \rceil$, is
    \begin{align*}
        \max_{i,j}|\mathsf{err}_{i,j}|
        &= \frac{ \log\lceil\tfrac{T}{w}\rceil +1}{\varepsilon}\sqrt{2d \ln \left(\frac{1.25}{\delta} \right)\ln \left(\frac{2 \lceil \tfrac{T}{w}\rceil \cdot dw }{\beta} \right)} \\
        &\leq  \frac{3\log \tfrac{T}{w}}{\varepsilon} \sqrt{d \ln\left( \frac{2T d}{\beta} \right) \ln\left( \frac{1.25}{\delta} \right) }\\
        &= \gamma.
    \end{align*}
    Adding everything together, it follows that, with probability $1-\beta$, the error in the $\lazycms$ is 
    \begin{align*}
        \hat{f}_a^{\lazycms} - f_a &\in [0, F_1/w] + [-w,0] + [-\gamma, \gamma],
    \end{align*}
    as required.
    Similarly, with probability $1- \beta$, the error in the $\lazycs$ is 
    \begin{align*}
        \hat{f}_a^{\lazycs} - f_a &\in \pm \left(\sqrt{\frac{3F_2}{w}} +w +\gamma \right).
    \end{align*}

    On each update, there are $d$ updates to $P$ and $d$ updates to $O$.
    Each update to $P$ is constant and an update to $O$ is the cost of updating $\gaussbm$ counter.
    On a stream of length $T^*$, this cost is $\mathcal{O}(\log T^*)$.
    Therefore, the update complexity is dominated by the cost of an update to $O$, which takes $\mathcal{O}\left(\log\tfrac{T}{\beta } \log \tfrac{T}{w}\right)$ time.
    Similarly, the memory cost is dominated by $O$.
    By Lemma~\ref{lem:gbm_utility}, the memory cost of $O$ is $\mathcal{O}\left(w\left(\log\tfrac{T}{\beta } \log \tfrac{T}{w}\right)\right)$.
    This completes the proof.
\end{proof}

%\section{Lazy Heavy Hitters}
%\label{sec:hh}
%\input{sections/heavy_hitters}

\section{Experiments}
\label{sec:experiments}

In this paper, we have proposed a novel differentially private linear sketch, designed for the continual observation model.
It is aimed at improving throughput relative to existing art. 
To evaluate this claim empirically, we have implemented a suite of data structures and performed a comprehensive analysis of their privacy, utility, and performance trade-offs. 

\subsection{Implemented Sketches}
The following sketches are implemented and evaluated for frequency estimation:
\begin{itemize}
    \item $\cms$ and $\cs$ serve as non-private baselines.
    \item The \textit{Lazy} variants ($\lazycms$, $\lazycs$).
    \item The \textit{Punctual} variants ($\punctualcms$, $\punctualcs$).
\end{itemize}
Unless stated otherwise, all sketches are implemented with a fixed depth of $d = 3$. 

\subsection{Datasets and Metrics}
\textbf{Synthetic Data.} 
To evaluate frequency estimation, we generate synthetic data streams following Zipfian distributions with a varying skew parameter.
Each synthetic stream consists of $2^{20}$ updates. The privacy parameter $\delta$ is fixed to $0.001$ across all experiments. 

We vary three key parameters in our experiments: 
The privacy parameter $\varepsilon$; 
the sketch width and sketch depth (memory);
and the skew of the input distribution (for synthetic data generation).
For evaluation we record:
\begin{itemize}
    \item {Runtime:} Total time taken to process all $2^{20}$ stream updates.
    \item {Accuracy:} Measured as \emph{Average Relative Error (ARE)} for the top-$15$ most frequent items in the stream:
    \[
    \mathrm{ARE} = \frac{1}{k} \sum_{i=1}^{k} \frac{|f_i - \hat{f}_i|}{f_i},
    \]
    where $f_i$ is the true frequency and $\hat{f}_i$ is the estimated frequency of the $i$-th heavy item.
\end{itemize}

\textbf{Real-World Data.} 
For heavy hitter detection, we use anonymized traffic traces from the CAIDA 2019 dataset\footnote{\url{https://www.caida.org/data/passive/passive\_2019\_dataset.xml}}.
The dataset consists of approximately $26.6$ million total source IP addresses and around $312{,}000$ distinct source IPs.
For all experiments we set $k=256$.
We vary two key parameters:
memory, by changing $\tilde{k}$ for $\lazyhh$ and changing the width for the $\cms$ (and keeping $k$ constant); and privacy, by changing $\varepsilon$.

For evaluation, we report:
\begin{itemize}
    \item {ARE:} For heavy hitters returned by each sketch.
    \item {Precision:} The proportion of reported items that are actual heavy hitters:
    \[
    \mathrm{Precision} = \frac{|\text{Reported} \cap \text{True}|}{|\text{Reported}|}
    \]
    \item {Recall:} The proportion of true heavy hitters that were reported:
    \[
    \mathrm{Recall} = \frac{|\text{Reported} \cap \text{True}|}{|\text{True}|}
    \]
\end{itemize}

%Note that Algorithm~\ref{alg:cos_hh} incorporates a thresholding mechanism to prevent items that might break privacy from being released. 
%While this threshold is theoretically polylogarithmic in the stream length, in practical settings it is essential that it remains low enough to allow meaningful detection of frequent items.
%Therefore, recall is a key metric in these experiments. 
%Our experiments verify that this is indeed the case for the CAIDA dataset.

\subsection{Experimental Environment}
All implementations are written in C++. 
Experiments were conducted on a Dell Latitude 7430 running Windows 11, equipped with an Intel Core i7 processor and 32GB of RAM. 
Each parameter configuration is evaluated over 20 independent trials, and we report averaged results.

\subsection{Synthetic Data}

\begin{figure*}[t]
    \centering
    \begin{subfigure}[b]{0.32\textwidth}
        \includegraphics[width=\textwidth]{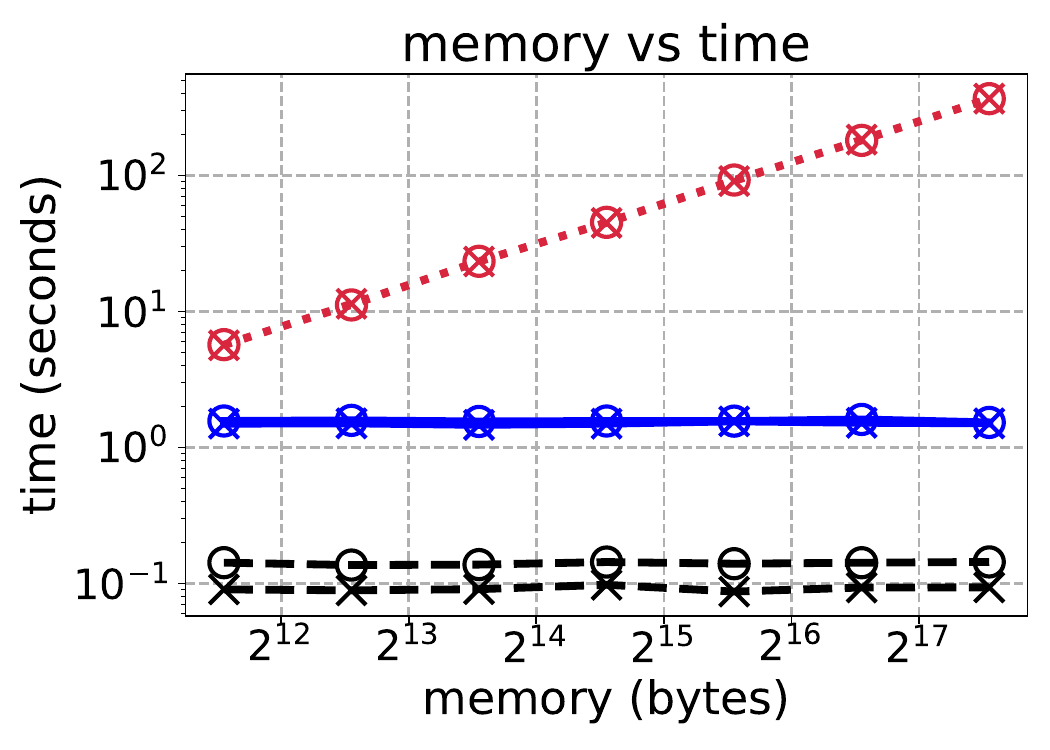}
        \caption{$\varepsilon = 0.3$, skew = 1.3}
        \label{fig:mem_freq_time}
    \end{subfigure}
    \hfill
    \begin{subfigure}[b]{0.32\textwidth}
        \includegraphics[width=\textwidth]{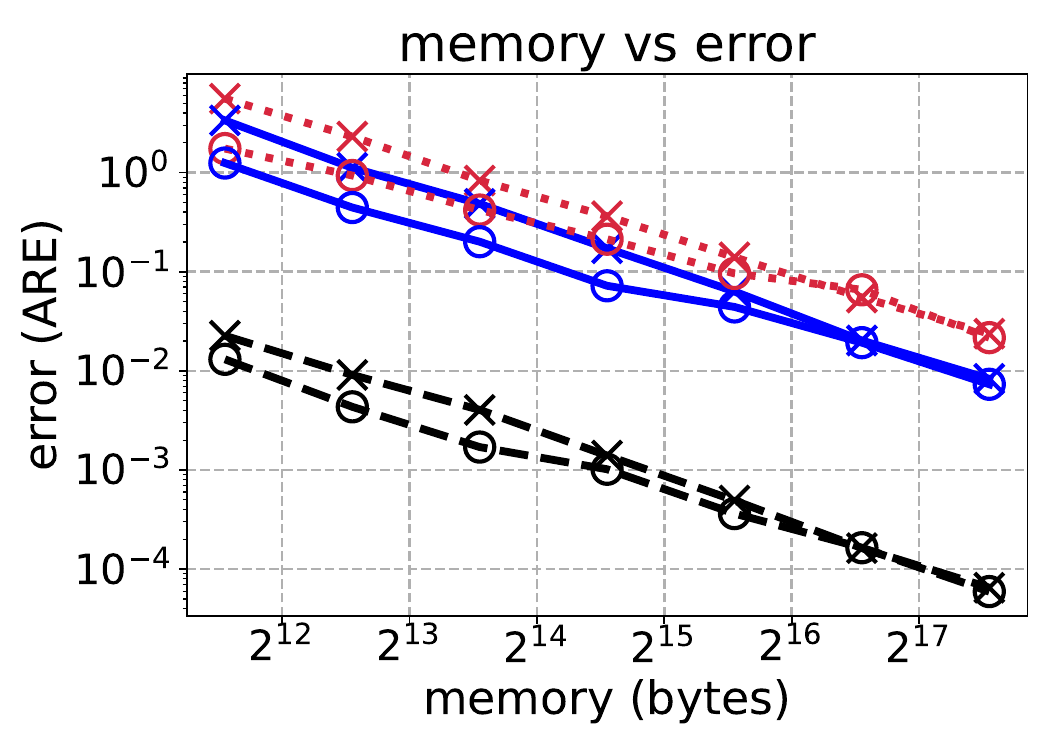}
        \caption{$\varepsilon = 0.3$, skew = 1.3}
        \label{fig:mem_freq_are}
    \end{subfigure}
    \hfill
    \begin{subfigure}[b]{0.32\textwidth}
        \includegraphics[width=\textwidth]{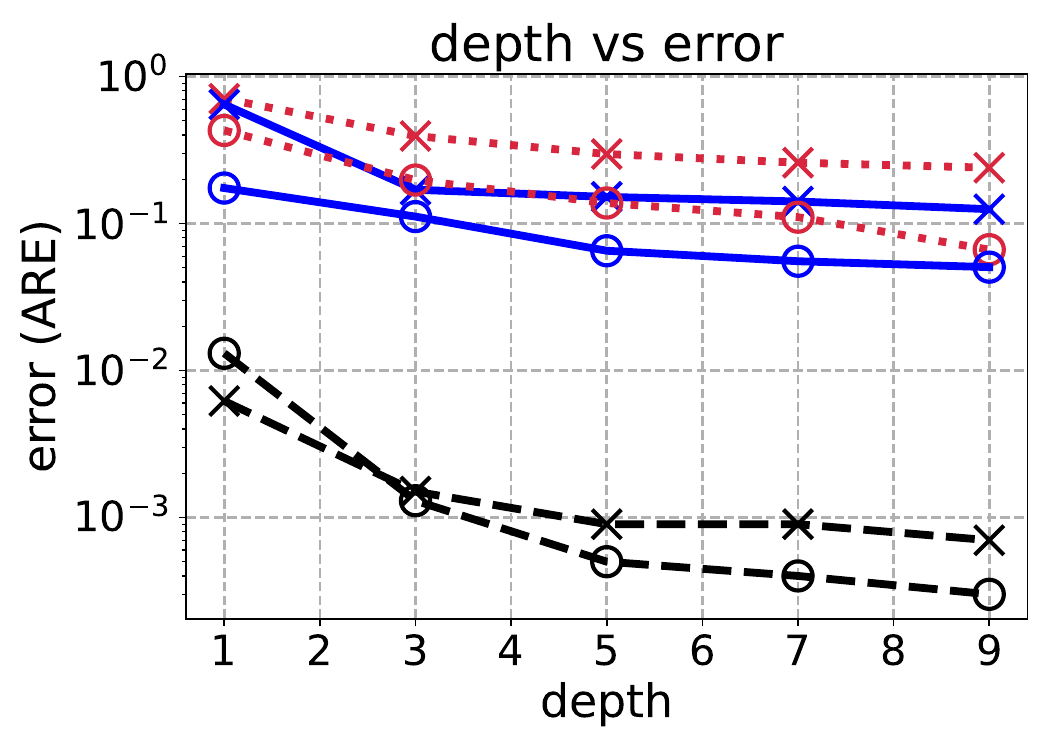}
        \caption{$\varepsilon = 0.3$, skew = 1.3}
        \label{fig:depth_freq_are}
    \end{subfigure}

    \vspace{0.5cm}

    \begin{subfigure}[b]{0.32\textwidth}
        \includegraphics[width=\textwidth]{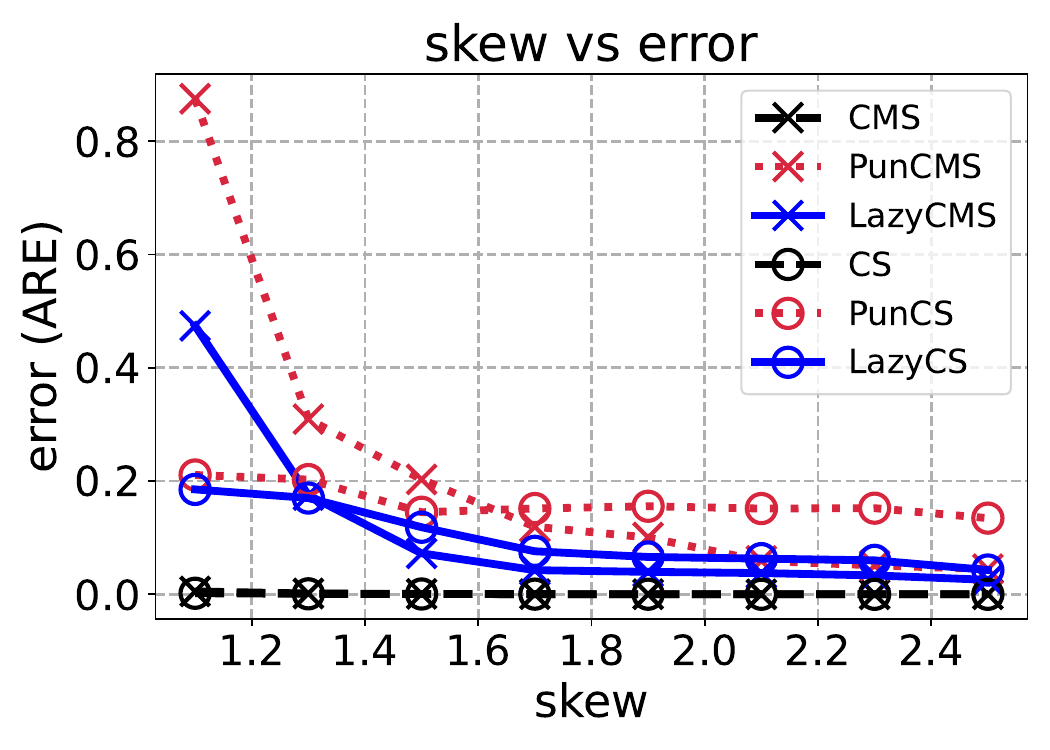}
        \caption{Memory = $24$KB, $\varepsilon=0.3$}
        \label{fig:skew_are}
    \end{subfigure}
    \hfill
    \begin{subfigure}[b]{0.32\textwidth}
        \includegraphics[width=\textwidth]{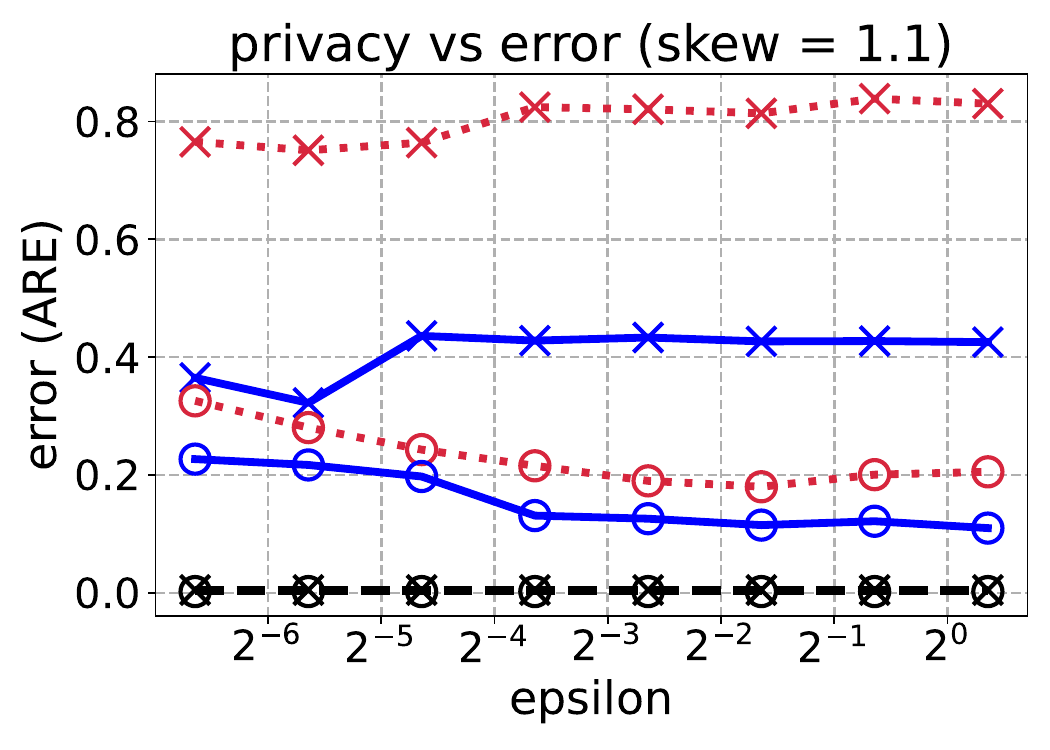}
        \caption{Memory = $24$KB, skew = 1.1}
        \label{fig:eps_are_11}
    \end{subfigure}
    \hfill
    \begin{subfigure}[b]{0.32\textwidth}
        \includegraphics[width=\textwidth]{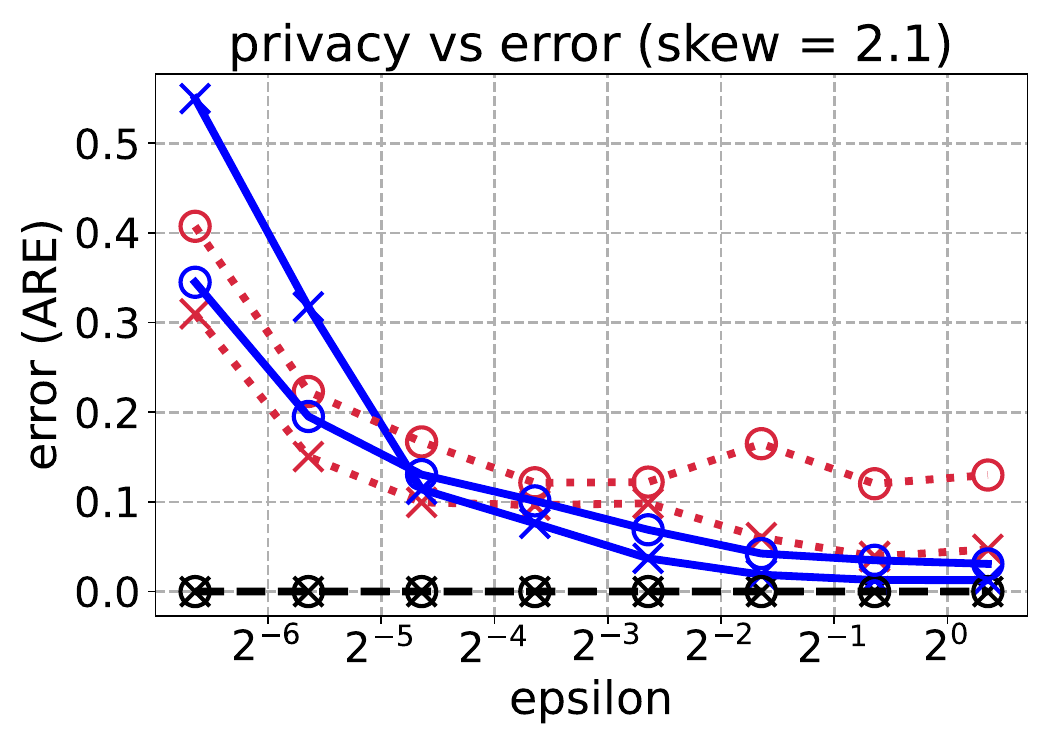}
        \caption{Memory = $24$KB, skew = 2.1}
        \label{fig:eps_are_21}
    \end{subfigure}

    \caption{Comparative frequency estimation on synthetic data. $\delta = 0.001$.
    For Figures~\ref{fig:mem_freq_time} \&~\ref{fig:mem_freq_are}, memory is increased by changing the width of the data structures.
    In contrast, for Figure~\ref{fig:depth_freq_are}, memory is increased by changing the depth of the data structures.
    }
    \label{fig:syn}
\end{figure*}

We begin by looking at frequency estimation.
A key contribution of our work is improving the throughput of existing work with the use of lazy updates.
We validate this contribution here, and observe additional benefits of our approach, such as improved utility.

Since the continual observation sketches implement the $\gaussbm$, each counter internally stores one node per level of the binary tree that overlays the stream.
This results in significantly higher per-counter memory usage than in traditional sketches. 
Consequently, for a fixed memory budget, relative to the non-private sketches, the sketch width (i.e., the number of columns) is lower for the private variants. 
Furthermore, in the $\lazy$, the effective stream length observed by each $\gaussbm$ counter is reduced to $T/w$.
This results in a smaller counter size compared with the $\punctual$ (See Table~\ref{tab:comparison}).
Thus, it has a larger sketch width under the same memory budget.
For instance, with a 24KB memory allocation, the widths for $\cms$, $\lazycms$, and $\punctualcms$ are 2000, 55, and 33, respectively.

Figure~\ref{fig:mem_freq_time} reports the total time to process the stream as a function of memory allocation (via changing the sketch width). 
As expected, non-private sketches ($\cms$, $\cs$) and the lazy variants ($\lazycms, \lazycs$) maintain constant throughput as the sketch width increases, since the update cost is independent of the sketch width. 
In contrast, the Punctual variants ($\punctualcms, \punctualcs$) exhibit a linear increase in update time as the sketch width increases.
This is because the $\punctual$ updates every column in each row of the sketch at each stream arrival. 
The $\lazy$ avoids this overhead with the use of lazy updates to the output, resulting in a significant improvement in throughput.
Notably, $\lazycms$ achieves up to $250\times$ higher throughput than $\punctualcms$, highlighting a key performance advantage of our approach.

Figure~\ref{fig:mem_freq_are} shows the relationship between sketch width and ARE on the top-15 heavy items. 
All variants demonstrate decreasing ARE with increasing memory, reflecting the reduction in hash collisions. 
However, the $\lazy$ consistently outperforms $\punctual$, achieving lower ARE at the same memory level. 
This is because $\lazy$ is initialized with a larger width to achieve the same memory allocation, which results in fewer hash collisions. 
The improved accuracy is an additional benefit of the $\lazy$, on top of its superior runtime performance.

Figure~\ref{fig:depth_freq_are} evaluates the effect of increasing memory (via sketch depth) on accuracy. 
All variants show decreasing ARE as depth increases, though the benefit is more pronounced for the non-private baselines. 
This suggests that adding more hash functions has diminishing returns in the private setting, where there is an additional source of error.

Figure~\ref{fig:skew_are} illustrates how the Average Relative Error (ARE) evolves with increasing input skew. 
As the distribution becomes more skewed, the ARE of the private sketches gradually levels off. 
This plateau reflects a shift in the dominant source of error: while higher skew reduces collision-induced errors in the sketch, the noise introduced by the differentially private counters remains constant, eventually becoming the primary contributor to the overall error.
%Skew has the largest effect on the non-private variants, particularly the $\cms$ baseline. 
%As skew increases, the frequency of the lower-ranked items among the top-15 heavy hitters decreases, making them more impacted by the noise added for privacy. 
%Though the influence of skew is still observable for the private sketches.

In Figure~\ref{fig:eps_are_11}, we evaluate ARE under a high skew ($\phi = 1.1$) while varying the privacy parameter $\varepsilon$. Here, decreasing privacy (i.e., increasing $\varepsilon$) has a limited effect on ARE. 
This is because the dominant error arises from collisions rather than noise. 
Interestingly, both $\lazycms$ and $\punctualcms$  exhibit lower ARE at very small $\varepsilon$ values. 
This can be attributed to the $\min$ query operation selecting a row with large negative noise, which counteracts collision-induced overestimation.
This is a form of noise-induced bias that occasionally benefits accuracy.

Figure~\ref{fig:eps_are_21} investigates the impact of $\varepsilon$ under low skew ($\alpha = 2.1$), where collision error is reduced and noise becomes the dominant error source. 
Here, we observe a clear decrease in ARE as $\varepsilon$ increases, consistent with the inverse dependence of noise magnitude on $\varepsilon$. 
Initially, ARE decreases approximately linearly with $\varepsilon^{-1}$, before saturating as noise becomes comparable to residual collision error. 

\section{Prior Work}
\label{sec:priors}

Differentially private sketches are motivated by the need to perform sensitive streaming analytics under formal privacy guarantees. 
These methods adapt classical sketching techniques to meet the requirements of differential privacy. 
In the 1-pass model, it is sufficient to inject a suitably calibrated noise into the sketch at the beginning of the stream.
Prior work has used Gaussian noise \cite{pagh2022improved,zhao2022differentially} and randomized response \cite{bassily2017practical}.
In the continual observation model, the current approach is to use a continual observation counter in each cell in the sketch \cite{epasto2023differentially}.
However, to maintain privacy, this requires updating every cell in the sketch at each stream arrival.
Reducing this computational overhead is the main motivation for this work.

Recently, differentially private sketches have seen increasing application in a range of privacy-preserving data analysis tasks, including synthetic data generation \cite{holland2025private}, quantile estimation \cite{zhao2022differentially}, and heavy hitter detection \cite{bassily2017practical}. 
These efforts demonstrate the importance of sketches as a foundation for scalable private analytics.

As an alternative to linear sketches, the $\mg$ algorithm~\cite{misra1982finding} can be privatized \cite{biswas2024differentially,chan2011private,lebeda2023better} and used for frequency and heavy hitter approximation.
$\mg$ is a counter-based method that maintains up to $w$ key-value pairs $(x, c)$, tracking frequent items and their approximate counts.

Unlike sketch-based methods, privatizing $\mg$ is non-trivial: simply adding noise to counters is insufficient because the \emph{identities} of tracked items can change between neighboring inputs, making the output distinguishable. 
In the 1-pass model, this issue is formalized by Lebeda and Tetek~\cite{lebeda2023better}, who show that $\mg$ can expose up to two \emph{isolated elements} (tracked items that differ between two neighboring streams). 
This implies that while isolated elements exist, the instability is limited and can be obfuscated with noise. 
Moreover, since these elements typically appear with very low counts (at most one), they can be suppressed using thresholding techniques (as in stability histograms), mitigating privacy leakage with high probability.

However, this idea cannot be easily extended to the continual observation setting.
The reason being that the identity of isolated elements can change as the stream progresses and an adversary (with access to neighboring datasets) can track these changes using the deterministic $\mg$.
Thus, the sensitivity of the outputs cannot be bound in any useful way.

\section{Conclusion}
\label{sec:conclusion}

We have presented a new framework for sketch-based frequency estimation that satisfies differential privacy under the challenging continual observation model. 
Our approach is built around a \emph{lazy update} mechanism, which selectively injects noise into a small, rotating portion of the sketch at each time step. 
This design choice dramatically reduces per-update computational overhead while maintaining strong utility and privacy guarantees.

Through the design and implementation of the $\lazy$ algorithm, we demonstrate that it is possible to achieve both high throughput and high utility in privacy-sensitive streaming applications. 
In particular, our experiments show that $\lazy$ improves update throughput by a factor of 250 compared to the state-of-the-art, while also offering better or comparable accuracy across a wide range of parameter settings.

\bibliographystyle{unsrtnat}
\bibliography{references}  %%% Uncomment this line and comment out the ``thebibliography'' section below to use the external .bib file (using bibtex) .

%\section{Additional Proofs}
%\label{sec:app_proofs}
%\input{sections/proofs}





% \end{thebibliography}

\end{document}